\documentclass[submission,copyright,creativecommons]{eptcs}
\providecommand{\event}{NCMA 2025} 

\usepackage{iftex}

\ifpdf
  \usepackage{underscore}         % Only needed if you use pdflatex.
  \usepackage[T1]{fontenc}        % Recommended with pdflatex
\else
  \usepackage{breakurl}           % Not needed if you use pdflatex only.
\fi
\usepackage{amsmath, amssymb,mathtools, mismath }
\usepackage{graphicx}
\usepackage{tikz}
\usetikzlibrary{automata, positioning}

\providecommand{\thisvolume}[1]{this volume of EPTCS, Open Publishing Association}

\newtheorem{defi}{Definition}
\newtheorem{thm}{Theorem}
\newtheorem{lem}{Lemma}

\newtheorem{prop}{Proposition}

\newenvironment{proof}{\begin{trivlist}
                      \item[]{\bf Proof}
                     \hspace{0cm} }{\hfill {\large $\bullet$}
                      \end{trivlist}}

\title{On some Classes of Reversible 2-head Automata}
\author{
Benedek Nagy
\institute{Department of Mathematics,
Eastern Mediterranean University\\
99628 Famagusta, North Cyprus, Mersin-10, Turkey\\ Department of Computer Science, Institute of Mathematics and Informatics,\\ Eszterh\'azy K\'aroly Catholic University, Eger, Hungary}
\email{nbenedek.inf@gmail.com}
\and
Walaa Yasin
\institute{Department of Mathematics,
Eastern Mediterranean University\\
99628 Famagusta, North Cyprus, Mersin-10, Turkey\\
}}
\def\titlerunning{On some Classes of Reversible 2-head Automata}
\def\authorrunning{Benedek Nagy \& Walaa Yasin}
\begin{document}
\maketitle

\begin{abstract}
Deterministic $2$-head finite automata which are machines that process an input word from both ends are analyzed for their ability to perform reversible computations. This implies that the automata are backward deterministic, enabling unique forward and backward computation. We explore the computational power of such automata, discovering that, while some regular languages cannot be accepted by these machines, they are capable of accepting some characteristic  
linear languages,
e.g., the language of palindromes.
Additionally, 
we prove that restricted variants, i.e., both 1-limited reversible 2-head finite automata and complete reversible 2-head finite automata are less powerful and they form a proper hierarchy.
In the former, in each computation step exactly one input letter is being processed, i.e., only one of the heads can read a letter. These automata are also characterized by putting their states to classes based on the head(s) used to reach and to leave the state.  
In the complete reversible 2-head finite automata, it is required that any input can be fully read by the automaton.
The accepted families are also compared to the classes generated by left 
deterministic linear grammars.
\end{abstract}

\section{Introduction}
Formal languages can be defined using grammars or recognized through various computational devices. The latter approach is particularly valuable for determining whether a specific word belongs to a given language. These devices can operate either deterministically or non-deterministically, with the latter often being more efficient for such tasks, e.g., in terms of complexity. However, this efficiency is often offset by the reduced expressive power of non-deterministic machines. 

Watson–Crick (WK) automata, introduced as part of DNA computing, combine automata theory with biological principles \cite{DNAComputing}. They operate on double-stranded tapes (representing DNA molecules), where each strand is read separately by read-only heads, maintaining the Watson–Crick complementarity relation. Various restricted classes of WK automata have been studied, with constraints on states and/or transitions.
A key concept in WK automata is the $5^{\prime}\to 3^{\prime}$ direction, which aligns with the biochemical reading direction of DNA strands. Some models use a sensing parameter \cite{OnhierarcyofsensingWatson}, which ensures that the reading heads are within a fixed distance or meet at the same position, enabling decision-making at that point. Research has shown that WK automata can characterize specific language classes, such as linear context-free languages and their subclasses (e.g., even linear languages \cite{iConcept}).
A newer model has been proposed to eliminate the sensing parameter while maintaining the same computational power \cite{NaCoNP22,AFL17}. It has been proven that its deterministic version accepts exactly the same class of languages (2detLIN) as the previous deterministic model with the sensing parameter \cite{10.1007/s00236-019-00362-6,UCNC18}. 
This class includes the class of even linear languages, which are of interest due to their learnability in formal languages \cite{SempereG94}.

Further studies explore related and expanded models, such as WK counter machines \cite{EgeciogluHN11,HegedusNE12}, WK automata with iterated reading \cite{Leupold}, double and 2-head jumping automata \cite{KocmanKM16,KocmanKMN22}, translucent letter models \cite{withtranslucent,lin-trans}, 2-head pushdown automata \cite{NCMA15,2headpushdown}, 2-head automata with output \cite{NCMA-NK19,NagyK21} and 2-head automata for circular words \cite{NCMA24neck}. These models refine the classification of languages and computational hierarchies, providing deeper insights into automata theory and DNA computing. We should also recall various concepts of determinism that are defined for 2-head automata: state-deterministic and quasi-deterministic $5^{\prime}\to 3^{\prime}$ WK automata \cite{state-det,quasi-det}, respectively. In computations in these models, not the next configuration, but the next state is defined uniquely based on the {current} state or on the {current} configuration, respectively.

The reversibility of standard Watson–Crick (WK) automata is an important aspect related to their computational behavior and efficiency. Since WK automata operate on double-stranded tapes with complementary base pairing, reversibility can be explored in terms of whether a computation can be uniquely reversed to reconstruct the original input \cite{ReversibleWatson} (with the condition that the initial state does not have any predecessor states). Now, concentrating on
 $5^{\prime} \to 3^{\prime}$ WK automata, the case is little bit different. Here, the computations follow the natural biochemical reading direction of DNA strands, and their deterministic variants ensure that each configuration leads to a unique subsequent configuration. However, for a WK automaton to be fully reversible, every transition must be invertible, meaning that for each computational step, there exists a unique predecessor configuration. This property is closely linked to bijective complementarity and sensing parameters, which influence how information is processed when the heads of the automaton meet. Investigating reversible WK automata (\cite{ReversibleWatson}) can provide insights into energy-efficient computations, as reversible computing is known to reduce energy dissipation, a concept relevant to both theoretical computer science and molecular computations.

In this research, the focus is on machines based on the concept of reversible deterministic 
2-head finite automata. Their non-deterministic variants are introduced in \cite{lin-auto,nagy2012class} and in \cite{NagyDNA07,AFL17}. 
A non-deterministic 2-head finite automaton consists of a finite control mechanism that reads symbols from a read-only input tape using two input heads. The left (first) head scans symbols from left to right, while the right (second) head moves from right to left. During computation, the finite control non-deterministically selects one or both of the heads, read(s) the next symbol, and transitions to a new state, with the new state chosen in a non-deterministic manner. Computation concludes when the two heads finish reading the input word and meet at some point on the tape. As is standard, the input word is accepted if there exists a computation that ends in a final state when the heads meet. 

Non-deterministic 2-head finite automata can recognize the class of linear context-free languages generated by such context-free grammars, where each production contains at most one non-terminal on the right-hand side. Additionally, there is a correspondence between 2-head finite automata and linear grammars, allowing for conversion between the two representations. The deterministic 2-head automata have less expressive power, the class 2detLIN is accepted by them \cite{10.1007/s00236-019-00362-6}. In this paper, we address the reversibility of the computations in these models.
Thus, we recall a related model from \cite{kutrib2023reversible}.
A reversible two-party Watson-Crick automaton (REV-PWK)  consists of two independent finite automata that process the same input string in opposite directions while communicating via messages to ensure synchronization and reversibility. 

In this paper, we consider special deterministic 2-head automata that are also backward deterministic, i.e., they are reversible. We also show that language families 2detLIN, and the classes accepted by reversible 2-head automata and their 1-limited and complete variants show a proper hierarchy, in which some of the new classes are incomparable  with the class of regular languages. Some of the main results are summarized in the diagram shown in the concluding section. 
Note that because of the page limit some of the proofs are omitted.

\section{Preliminaries and Basic Definitions} 
This section reviews some of the 
fundamental concepts in formal language and automata theory. We assume that the reader has prior knowledge of these fundamental concepts in this field.  
Otherwise, they may refer to sources such as \cite{Hopcroft1979IntroductionTA,book} for further details.

Let \( V \) be a finite, non-empty set, an \textit{alphabet} of symbols, commonly referred to as letters. Sequences formed using these letters are known as \textit{words}. A collection of such words constitutes a language over the alphabet \( V \). The symbol \( \lambda \) represents the empty word.
  
To establish our notation, we recall that a \textbf{linear context-free grammar} is a generative grammar defined as $(N, V, S, P)$, where $N$ and $V$ represent the sets of nonterminal and terminal symbols, respectively, $S$ is the start symbol ($S\in N$), and $P$ is the set of production rules. Each production rule (or rewriting rule) follows the form $A\to u$, where $A\in N$ and $u \in V^{*}+V^{*}NV^{*}$. The class of linear context-free grammars generates the family LIN, which consists of all linear context-free languages. In general, for any generative or accepting system $A$, we denote its generated or accepted language as $L(A)$.

\subsection{2-Head Finite Automata Accepting Linear Languages and their Variants}

We recall a variant of finite automata with two heads based on \cite{lin-auto,nagy2012class} that plays essential roles in this paper. They read the word from the beginning and the end, in parallel, in opposite directions.

\begin{defi}
A non-deterministic 2-head automaton is a 5-tuple $( Q, q_{0}, V, \delta, F )$ with the transition function 
\[
\delta : Q \times (V \cup \{\lambda\}) \times (V \cup \{\lambda\}) \to 2^Q,
\]
where $Q$ is the finite set of states, $q_{0} \in Q$ is the starting (or initial) state, $F \subset Q$ contains the final states, and $V$ is the %set of
alphabet.  
\end{defi}
\smallskip

Computations and accepting computations are defined through configurations containing the {current} state and the part of the input that has not been processed yet:
 \\
When input $w$ is provided, the computation begins with $(q_0, w)$. A computation step  $(q, aw^\prime b)\Rightarrow (p,w^\prime)$  can occur if $p\in \delta(q,(a,b))$. The automaton processes the entire input word until the heads meet. Each letter is read by one of the heads during an accepting computation. The word $w$ is accepted if $(q_{0}, w)\Rightarrow ^{*} (q, \lambda)$ where $q\in F$.

In the graph of an automaton in transitions, each arrow is labeled with a pair of symbols $(a, b)$, indicating that the first head reads symbol $a$ and the second head reads symbol $b$, with both heads moving in one step. We allow any or both $a$ and $b$ to be $\lambda$. Additionally, we use the notations $\rightarrow a$ to represent $(a, \lambda)$ and $\leftarrow a$ to represent $(\lambda, a)$. 
As, actually, only those states play role in any computations that can be reached from the initial state, we may assume in the followings that all states of the used automata are reachable, i.e., there is an input such that the/a computation on this input includes a configuration containing the given state.  
Transitions when both heads read the empty word can easily be eliminated (somewhat similarly as $\lambda$-transitions from a traditional 1-head nondeterministic finite automaton); moreover, they are not allowed in deterministic variant (similarly as there is no $\lambda$-transition in the case of traditional deterministic finite automata).

\begin{defi}
A 2-head finite automaton is deterministic if for every possible configuration, there is at most one transition step that can occur and at least one of the heads of the automaton reads an input letter in this transition (if any). In other words, a 2-head automaton is deterministic if and only if for every word $ w \in V^*$ and every state $ q \in Q$ there is at most one pair $w^{\prime} \in V^*$ and $q^{\prime} \in Q$ such that the transition $(q,w) \Rightarrow (q^{\prime},w^{\prime})$ is valid and $w\ne w'$.
\end{defi}

The deterministic version of these automata is weaker than the general, non-deterministic variant: i.e., they do not accept all linear languages. Consider the linear language $L_{1\lor3}=\{a^n b^n\} \cup \{a^{3n} b^n\}$ ($n > 0$). It is clear that it can be accepted by a non-deterministic 2-head finite automaton trying both possibilities to check in a non-deterministic way. 
For a deterministic automaton, informally, it should be decided which head moves in which step. With a finite control, it is impossible to know at first how many steps of the first head should be followed by a step of the second head. For more details on the model, see e.g., \cite{10.1007/s00236-019-00362-6}.  

\begin{thm}
A 2-head automaton $(Q, q_{0}, V, \delta, F)$ is deterministic if and only if
the following conditions are satisfied:
\begin{enumerate}
\item For each $q\in Q$, $\delta(q,(\lambda, \lambda))= \emptyset$.
\item For each $q \in Q$ and $a,b \in V\cup \{ {\lambda}\}$, $ab \ne \lambda$ we have $|\delta(q, (a,b))| \leq 1$.
\item If $\delta(q, (a, \lambda))\neq \emptyset$, where $a\in V$ and $\delta(q, (c,d))\neq \emptyset$, then $c \neq a$ and $c\ne \lambda$.
\item If $\delta(q, (\lambda, a))\neq \emptyset$, where $a\in V$ and $\delta(q, (c,d))\neq \emptyset$, then $d \neq a$ and $d\ne \lambda$.
\end{enumerate} 
\end{thm}

Deterministic 2-head finite automata are denoted as D-2H, while
the class of all languages recognized by them is denoted by 2detLIN \cite{NCMA25-MN2head,10.1007/s00236-019-00362-6,UCNC18}.
Moreover, specific variants of these automata are also defined and used \cite{CiE2009,Annales,AFL17,NaCoNP22,UCNC18}. %
In \emph{1-limited} variant, in each computation step exactly one of the heads is reading an input symbol. On the other hand, we may also define ``complete'' variant that is somewhat analogous to the completely defined deterministic finite automata used in the case of regular languages. We say that a D-2H automaton is \emph{complete} if for every possible configuration $(q,w)$ with a nonempty word there is exactly one pair $w^{\prime} \in V^*$ and $q^{\prime} \in Q$ such that the transition $(q,w) \Rightarrow (q^{\prime},w^{\prime})$ is valid.
It is also clear (as it is proven in \cite{10.1007/s00236-019-00362-6}) that every language in 2detLIN is accepted also by a 2-head automaton that is deterministic, 1-limited and complete. 

As our main topic is reversibility, we give our first new definition. Remember that we assume that all states of the automata we consider are reachable, and therefore, the set of all possible configurations is the same as the set of  configurations that appear in some computation. 

\begin{defi} 
A 2-head finite automaton $A$ is backward deterministic if at each possible configuration in a computation on a given input there is at most one predecessor configuration in the computations of $A$, i.e., in each configuration it is clear what was read in the last transition (by knowing what parts of the original input have already been processed) and from which state we arrived to the state of the {current} configuration.
\end{defi}

\begin{thm}
A 2-head finite automaton $A$ is backward deterministic if and only if $\forall w^{\prime} \in V^*$ and $\forall q^{\prime} \in Q$ and $\forall a,b\in V$ there exists at most one 
$w \in V^*$ with $w = aw'$, $w=w'b$ or $w=aw'b$
and at most one $q \in Q$ such that $(q,w) \Rightarrow (q^{\prime},w^{\prime})$.
\end{thm}
\begin{proof}
For the first direction, for every $ w' \in V^* $, $ q' \in Q $, and $ a,b \in V $, there is at most one $ w \in V^* $ such that $ w = aw' $, $ w = w'b $, or $ w= aw'b $, and at most one $ q \in Q $ such that $ (q, w) \Rightarrow (q', w') $. This means that for any configuration $ (q', w') $, there is at most one possible previous configuration $ (q, w) $ leading to it in one step if either $a$ is read by the first head or $b$ is read by the second head, or both at the same time.
Therefore, the reverse transition relation (from $ (q', w') $ to its predecessors $ (q, w) $) is deterministic: each configuration has at most one predecessor. This is precisely the definition of backward determinism.

For the second direction, suppose $ A $ is backward deterministic, i.e., for every  configuration $ (q', w') $, there is at most one configuration $ (q, w) $ such that
$(q, w) \Rightarrow (q', w')$ knowing what was already read from the input with each head.

We now want to show that this implies the stated condition: Let us suppose, for contradiction, that for some $ w' \in V^* $, $ q' \in Q $, and $ a, b \in V $, there exist two different pairs $ (q_1, w_1) $ and $ (q_2, w_2) $ such that
$ w_1,w_2 \in \{aw' , w'b, aw'b\} $ and at least one of $q_1 \ne q_2$ and $w_1\ne w_2$ holds. 
 However, in this case, both
$ (q_1, w_1) \Rightarrow (q', w') $, and $ (q_2, w_2) \Rightarrow (q', w') $ are valid computation steps, therefore $(q', w') $ has at least two predecessor configurations, contradicting backward determinism.

So, the assumption that more than one such $w$ and/or $ q $ exist leads to a contradiction. Therefore, the condition must hold.
\end{proof}

Now we turn to reversible 2-head finite automata, which is our main topic here. Basically, reversibility in finite automata is meant with respect to the possibility of stepping the computation back and forth. So, the machine has also to be backward deterministic.

For classical finite automata, an automaton for the reversal of the accepted language can be obtained by constructing the reversal \cite{rabin1959finite}, or dual automaton \cite{onreversibleautomata}, i.e., by reversing the transitions and interchanging initial and final states. For 2-head automaton, we can obtain an automaton for the reversal of the language by simply interchanging the role of  the first and second heads in each transition.

\begin{lem}\label{reversal}
Let $A=(Q, q_0, V, \delta, F)$ be a 2-head finite automaton that recognizes a language $L$. Then, there exists a 2-head finite automaton $A^{\prime} = (Q, q_{0}, V, \delta^{\prime}, F)$ that recognizes the reversal of $L$ (denoted as $L^R$), which can be obtained by interchanging the transitions of the first and second heads in $A$, i.e., $\delta^{\prime}(q,(a,b))= \delta(q,(b,a))$ for all $q \in Q$ and $a,b\in V\cup \{\lambda\}$.
\end{lem}

\begin{proof}
Since a 2-head automaton reads an input string using two separate heads, reversing the language requires adjusting how these heads process the input. By modifying the transition function such that the roles of the two heads are swapped, the automaton can effectively process the reversed string in the same manner as the original automaton processed a word of $L$. This ensures that $A^{\prime}$ definitely accepts $L^R$ while preserving the computational structure of $A$.
\end{proof}

Thus, we can see that the reversal of the language is not connected to backward transitions. However, it is also interesting to see what happens, if we apply a similar construction as for finite automata to 2-head finite automata.

\begin{defi}\label{def-rev}
 A 2-head automaton is reversible (it is \textnormal{R-2H}, for short), if it is both deterministic and backward deterministic.   
\end{defi}

\begin{thm}\label{char-rev}
A D-2H automaton is reversible, if for any 2 transitions $\delta(q_{1},(a,b))=q_{2}$ and $\delta(q_{1}^{\prime},(c,d))=q_{2}$, then
$a\neq c$ or $b\neq d$ with $|ad|\geq 1$ and $|bc|\geq 1$. 
\end{thm}
\begin{proof}
To prove that a deterministic $2$-head automaton is reversible under the given conditions, we must establish that every transition in the automaton is uniquely reversible.
Suppose that there exist two transitions: $\delta(q_1, (a, b)) = q_2$, and $\delta(q_1^{\prime}, (c, d))= q_2$, we must ensure that the original states $q_1$ and $q_1^{\prime}$ can be uniquely determined from $q_2$ by knowing what was the original input.
The given condition states that at least one of 
$a \neq c$ and $b \neq d$ holds. This ensures that no two transitions map 
the read input pairs $(a,b)$ and $(c,d)$ to the same state unless at least one symbol differs. This avoids ambiguity in the reverse of the transition.
Additionally, the condition $|ad| \geq 1$ and $|bc| \geq 1$  ensures that the left head in the first transition and the right head in the second transition cannot be $\lambda$ at the same time, and the same applies to the right head from the first transition and the left head from the second transition. This means at least one of the heads reads a nonempty input, a letter, in the transition, and it is decided which head moves in which step in the reversal, ensuring that reversing the transition is always possible.
\end{proof}

\begin{lem}\label{reg-notrev}
The regular language $ L_{ab}= \{a^n b^m \mid n, m \geq 0\} $ cannot be accepted by \textnormal{R-2H}.
\end{lem}
\begin{proof}
We aim to prove that the language $
L_{ab} = \{a^n b^m \mid n, m \geq 0\} $ cannot be accepted by any reversible 2-head automaton. A reversible automaton must satisfy the property that every configuration has at most one predecessor and one successor, ensuring unique invertibility of transitions.

Assume that there exists a reversible 2-head automaton $ M = (Q,\{a,b\}, q_0, \delta , F)$ that accepts the language $ L_{ab} = \{a^n b^m \mid n, m \geq 0\} $. Let $|Q| = k $, where $ k $ is finite. The automaton $ M $ must process the input string $ a^n b^m $ while maintaining reversibility. We consider the three possible cases for how $ M $ might read the input $a^nb^m$ ($n,m>>k$) in the first steps of the computation:
\begin{enumerate}
\item Reading $a$'s and $b$'s together. Suppose $ M $ uses transitions like $p = \delta(q_0, (a, b)) $ to read the first $ a $ and the last $ b $ simultaneously with a state $p$ (maybe it equals to $q_0$). However, as also the word $a \in L$, it must be accepted, thus there must be a transition $q' = \delta(q_0,(\lambda,a))$, since the transition $q'\in \delta(q_0,(a,\lambda))$ would contradict to the fact that $M$ is deterministic. Also, it is clear that $q' \ne q_0$, otherwise input $a^kb^ba$ would also be accepted. On the other hand $b\in L$ must also be accepted, thus there must be a transition from the initial state to accept it. However, neither $q''\in \delta (q_0,(b,\lambda))$, nor $q''\in \delta (q_0,(\lambda,b))$ could be, as any of those would contradict to the determinism of $M$ with the previously described two transitions. Thus, to process the word $a^nb^m$ with $n,m>>k$, the first computation step cannot process both an $a$ and a $b$. 
\item Reading only an $ a $ first. Then there is a transition $p\in \delta(q_0,(a,\lambda))$ (maybe with $p=q_0$) in $M$.
However, $M$ must also accept every word from $b^*$, thus it must also have a transition $q'\in \delta(q_0,(b,\lambda))$ (since $M$ is deterministic, it cannot have a transition $q'\in \delta(q_0,(\lambda,b))$). If $q' = q_0$, then $M$ would also accept a word of the form $b a^n b^m$ that leads to a contradiction. Thus, $q' \ne q_0$ must hold. In this case, however, as all other words of $b^*$ must be accepted, there is an accepting continuation of the computation from the configuration $(q', b^m)$ with $m> k$.
As $m$ is larger than the number of states, there is a state $q''$ such that from $(q'', b^j)$ the computation also reaches the configuration $(q'',b^i)$ with $i< j\leq m$, i.e., $M$ has a cycle reading only $b$-s. Let $q''$ be the state for which the value of $j$ is the maximal (among states included in the cycle). Then, there are at least two different transitions, i.e., transitions from two distinct states  to $q''$ by reading only a $b$. 
However, in this case $M$ cannot be reversible, causing a contradiction. (Note that $q''$ could be the same as $q'$, but it cannot be $q_0$.)
%as at state $q''$  
\item Reading a $b$ first by the transition $p\in \delta(q_0,(\lambda,b))$. 
  The proof of this case is symmetric to the previous case, by interchanging the roles of left and right heads and the letters $a$ and $b$, leading again to contradictions.
\end{enumerate}
In all cases, the finite set of states %space 
of $ M $ and the requirement of reversibility create fundamental limitations. 
Cycles in the automaton 
violate reversibility by making it impossible to uniquely reconstruct the history of the computation. The deterministic nature of $ M $ further restricts its ability to handle transitions without ambiguity.

Thus, no reversible 2-head automaton can accept the regular language $ L = \{a^n b^m \mid n, m \geq 0\} $.
\end{proof}

Based on the fact that each regular language is in 2detLIN and  the previous Lemma, we can state our first hierarchy result.  

\begin{thm}
The class of languages accepted by reversible 2-head automata is a proper subset of the class accepted by deterministic 2-head automata.
$$2revLIN \subsetneq 2detLIN .$$
\end{thm}

\begin{proof}
The inclusion comes directly from the definition, as each reversible 2-head automaton is also deterministic. The properness is a consequence of Lemma \ref{reg-notrev}.
\end{proof}

In the next sections we consider some specific variants and show that they can accept only subclasses of 2revLIN.
However, first, in the next subsection, some related models are recalled.

\subsection{Related Models}

First, we recall a related model from \cite{kutrib2023reversible} that we have already mentioned.

A reversible two-party Watson-Crick automaton (REV-PWK)  consists of two independent finite automata that process the same input string in opposite directions, communicating via messages to ensure synchronization and reversibility. There exist reverse transition functions $ \delta_i^- $ for each automaton and reverse broadcast functions $ \mu_i^- $. These functions ensure that every configuration has at most one predecessor, which can be computed using another two-party Watson-Crick system. The upper component ($ A_1 $) reads the input tape from left to right, while the lower component ($ A_2 $) reads the input tape from right to left. During a forward computation step, the upper component reads an input symbol and then moves its head to the next position, and the lower component moves its head first and then reads the input symbol. During a backward computation step, the reverse behavior occurs, the upper component either moves its head to the left or stays stationary, while the lower component either moves its head to the right or stays stationary. Two-party Watson-Crick automata are generally complex due to the need for communication and coordination between two independent components. This communication mechanism allows the automata to coordinate their actions, making the model more expressive than systems without communication. In contrast, a reversible 2-head automata, which uses a single automaton with two heads reading the input in opposite directions, are simpler because they involve a single automaton controlling both heads, as they lack the communication capability, limits its ability to handle certain languages. For instance, the language $ \{a^n b^m \mid n, m \geq 0\} $ can be accepted by a REV-PWK due to its synchronization abilities but cannot be accepted by our reversible 2-head automata as we have already shown.

We also recall other models that use the generative approach.

The definition of deterministic linear grammars, as referenced in \cite{LRDet,Ldet}, specifies two key properties. First, the left deterministic linear grammar does not include production rules where the right-hand side begins with a nonterminal, Additionally, for any given nonterminal $T$, the first terminal appearing on the right-hand side of a rule with $T$ as the left-hand side uniquely identifies the rule, and this terminal is always followed by a nonterminal. While the right deterministic linear grammar does not include production rules where the left-hand side begins with a nonterminal, and for any given nonterminal  $T$, the first terminal appearing on the left-hand side of a rule with $T$ as the rigt-hand side uniquely identifies the rule, and this terminal is always followed by a nonterminal.

We recall the definition of the former special class of linear languages which were used in \cite{LRDet,Ldet}, and show their relation to our reversible 2-head automata definition.

\begin{defi}\label{def-left}{\textbf{(Left Deterministic Linear Grammar)}}
A left deterministic linear grammar (LDLG) 
$$G = (N, V, S, P)$$
is a linear grammar where all rules are of the form $T\to a T^{\prime} u$ or $T \to \lambda$, where for each $T, T^{\prime}, T^{\prime \prime}\in N$, $u, v \in V^*$ and $a \in V$, if $T\to a T^{\prime}u, \text{~and~} T \to a T^{\prime \prime}v$ are in $P$, then $T^{\prime}= T^{\prime \prime} \text{~and~} u=v.$
\end{defi}

The languages generated by LDLG 
are called left deterministic linear language (LDLL). 

Note that, in fact, the symmetric form called right deterministic linear grammars/languages (RDLG/ RDLL) are also defined in \cite{LRDet,Ldet}. 
Due to the symmetry of these grammars RDLG and LDLG,  the reversal of any language of one of these classes is in the other class.
Because of lack of space we do not go into further details about them.

\section{Reversible Automata Consuming the Input Letterwise}

In this section, we consider reversible 2-head automata with limitation on the number of heads that can move.
Thus, we continue by defining a specific subset of the class R-2H automata. We have already mentioned 1-limited variants of 2-head automata, now we define formally also for the reversible case. 

\begin{defi} 
Let $A$ be a (deterministic/reversible) 2-head automaton. 
If $A$ has the property that
every transition is of the form $\delta(q,(a,\lambda))$ or every transition is of the form $\delta(q,(\lambda,b))$ ($a,b\in V$), then  $A$ is a 1-limited (deterministic/reversible) 2-head automaton.

The class of 1-limited reversible 2-head automata is denoted by R1-2H,
while the class of languages accepted by 1-limited reversible 2-head automata is denoted by 2rev1LIN. 
\end{defi}
In fact the above definition states that for each state $q$ of the automaton for any (not necessarily distinct) letters $a,b\in V$, $\delta(q,(a,b)) = \emptyset$ always hold.

Now let us consider these 1-limited reversible 2-head automata, where in each computation step exactly one input letter is being processed. We can classify the states of the automaton as 
Table \ref{tab:stateClass} shows the possible properties. In fact, we can state this classification as a characterization result. 

\begin{table}
    \caption{Classifying states in a reversible 1-limited 2-head automata. The first index indicates which head was used to reach the given state, the second index shows which head is allowed to read in transition(s) from the given state, where $\to$ represents the first, $\leftarrow$ represents the second head and $\emptyset$ represents the case when there is no incoming or outgoing transition from/to the state (which may happen at an initial or a final state, respectively) The trivial one-state automaton without any transitions is not considered.}
    \label{tab:stateClass}
    \centering
    \begin{tabular}{l|cc|l}
  to $ \ \ \ \ \ \ \ \ \ \ \backslash\ \ $ from $q$  & first head & second head & \ \ NO HEAD \\ \hline
  first head  &  $Q_{\to,\to}$   &  $Q_{\to,\leftarrow}$ & $Q_{\to,\emptyset}$ \ \ \ (final\\
  second head       & $Q_{\leftarrow,\to}$   &  $Q_{\leftarrow,\leftarrow}$ & $Q_{\leftarrow,\emptyset}$ \ \ \ \ state)\\ \hline
    NO HEAD       & $Q_{\emptyset,\to}$ & $Q_{\emptyset,\leftarrow}$ & \\ 
    & (inital & state) &\\
    \end{tabular}
 \end{table}

\begin{thm}\label{thm:character}
Let $M$ be a 
1-limited reversible 2-head automaton  
accepting a nonempty language. Then its set of states can be written as the union of at most 7 disjoint classes according to Table \ref{tab:stateClass}. 
\\ The initial state $q_0$ can be any of the following 6 classes: $$ Q_{\to,\to},  Q_{\to,\leftarrow},  Q_{\leftarrow,\to},  Q_{\leftarrow,\leftarrow}, Q_{\emptyset,\to} , Q_{\emptyset,\leftarrow},$$ depending on if $q_0$ can be part of any configurations other than starting ones. Each accepting state can be in of the following classes:  $$Q_{\to,\to}, Q_{\to,\leftarrow}, Q_{\to,\emptyset}, Q_{\leftarrow,\to},  Q_{\leftarrow,\leftarrow}, Q_{\leftarrow,\emptyset}.$$ Every other state $q\in Q\setminus F\setminus\{q_0\}$ can be in one of the four classes:
$$ Q_{\to,\to},  Q_{\to,\leftarrow},  Q_{\leftarrow,\to},  Q_{\leftarrow,\leftarrow}.$$
\end{thm}

\begin{proof} We assume that all the states of $Q$ are reachable and useful. Thus, let us start the investigation by ordinary states, i.e., if $q\in Q\setminus ( F\cup \{q_0\})$.
  Then, as $q$ is reachable there is a transition of the form $q\in\delta(p,(c,d))$, where exactly one of $c$ and $d$ is a letter of the alphabet. To make sure that $M$ is reversible, then all transitions of the form $q\in\delta(p',(e,f))$ must also be a similar type, i.e., the same head must be used. Moreover, if the read letter is the same, than the same state $p=p'$ must be, otherwise $M$ would not be backward deterministic.
   Thus, the direction of the first arrow in the index, is already specified by this/these transition(s). Since $q$ is useful, there is an accepting computation that contains it in a configuration. Thus, there are also transitions from $q$: there is a state $r \in\delta(q,(a,b))$ with some pairs such that either $a$ or $b$ is an input letter. However, then, each transition from state $q$ must use the same head for reading, otherwise $M$ would not be deterministic. Thus, the second arrow direction is also fixed.
   Notice that in the same automaton all the four possible states may occur.

   Now, considering the initial state, if there is a transition to it, then exactly one of the above conditions apply. 
      However, if there is no transition to the initial state, then depending on which head is allowed to read in transitions from the initial state, it is in one of the $Q_{\emptyset,\to} , Q_{\emptyset,\leftarrow}$ categories. However, no other reachable state can be in any of these categories, therefore, at most one of these categories may appear in an automaton $M$, but not both at the same time.

      Finally, considering any of the accepting states, let us say $q\in F$: either there is at least one transition defined from $q$, and therefore it is in one of the four standard categories; or if no transition is defined from $q$, then it can be in one of the  $Q_{\to,\emptyset},  Q_{\leftarrow,\emptyset}$
      categories. Notice that, as $M$ may have many accepting states, in an automaton both of these special categories may occur.

 Thus, we can see, that $M$ as a 1-limited reversible 2-head automaton, can have at most 7 disjoint classes of states. 
\end{proof}

Of course, the possible transitions are also constrained based on the categories of the states.
From a state where the second arrow direction is $\ell \in \{\rightarrow,\leftarrow\}$ all transitions are going to states where the first arrow is also $\ell$. 

To make complete the characterization, we also state the following:

\begin{thm}
    A 1-limited 2-head automaton $M$  is reversible if and only if it has the following properties.
    \begin{itemize}
\item The states of $M$ can be classified according to the classes of Theorem \ref{thm:character}.
\item $M$ is deterministic: for each state $q$ and each letter $a$, there is at most one transition, i.e., \\ $|\delta(q,(a,\lambda))| \leq 1$ for states where the first head is allowed to read, and $|\delta(q,(\lambda,a))| \leq 1$ for states where the second head is allowed to read.  
    \item $M$ is backward deterministic, i.e., for each state $q$ and each letter $a$, there is at most one transition to arrive to $q$, i.e., there is at most one state $p$ such that $q \in \delta(p,(a,\lambda)) $ for state $q$ that can be reached by transition(s) in which the first head is allowed to read, and there is at most one state $p$ such that $q \in \delta(p,(\lambda,a)) $ for a state $q$ that can be reached by transition(s) in which the second head is allowed to read.
    \end{itemize}
\end{thm}
\begin{proof}
This theorem is the consequence of the definition and the previous characterization result.    
\end{proof}

Now, we intend to show that although the class of 1-limited variant of D-2H is able to accept all languages in 2detLIN, for reversible automata the 1-limited variant cannot be used as a normal form, as they are weaker than the variants without such restriction.

\begin{lem}\label{lem:wcbn}
    The language $L_{wcb^n}=\{wcb^n~|~w\in\{a,b\}^*,~|w|_b = n\}$ is not accepted by any R1-2H automaton. 
\end{lem}

We can now conclude the following hierarchy result:
\begin{thm}
    The languages accepted by R1-2H is a proper subclass of the languages accepted by R-2H: $$2rev1LIN \subsetneq 2revLIN .$$ 
\end{thm}
\begin{proof}
The inclusion comes directly from the definitions. The properness can be proven, on the one hand,
considering the R-2H automaton with two states as follows: at the initial state $q_0$ let there be two loop transitions: by $(b,b)$ and by $(a,\lambda)$. Further, let a transition by $(c,\lambda)$ from $q_0$ to $q_f$ which is the final state. Easy to see that this automaton is reversible.
This R-2H automaton accepts the language $$L_{wcb^n}=\{wcb^n~|~w\in\{a,b\}^*,~|w|_b = n\}.$$
On the other hand, as it was shown in Lemma \ref{lem:wcbn} this language is not in 2rev1LIN. 
\end{proof}

\section{Completely defined Reversible 2-Head Automata}  

In this section we investigate another subset of R-2H based on the already mentioned ``complete'' restriction.

\begin{defi}
A deterministic/reversible 2-head finite automaton is complete if for every possible configuration, there is exactly one transition step that can occur, if the input is not yet fully processed. In other words, a deterministic/reversible 2-head automaton is complete if and only if for every nonempty word $ w \in V^* \setminus \{\lambda\}$ and every state $ q \in Q$ there is exactly one pair $w^{\prime} \in V^*$ and $q^{\prime} \in Q$ such that the transition $(q,w) \Rightarrow (q^{\prime},w^{\prime})$ is valid.

The class of reversible complete 2-head automata is denoted by RC-2H, while the class of languages accepted by them is denoted as 2CrevLIN.
\end{defi}

Now, on the one hand, in \cite{10.1007/s00236-019-00362-6} it is proven that each language of 2detLIN can be accepted by a complete deterministic 2-head automaton. (Actually, for each language of 2detLIN there is also a  complete 1-limited deterministic 2-head automaton that accepts it.)
On the other hand, our aim, in this section, is to show that the property ``completeness'' is a real restriction for reversible 2-head automata.

\begin{prop}\label{RC&R1}
 There is no complete reversible 2-head finite automaton that has any transition of the form: $\delta(q, (a, b)) = q'$ {for some } $a, b \in V$.
 That is, no complete reversible 2-head automaton can have any transition where both heads move simultaneously.
 \end{prop}
\begin{proof}
Let $A= (Q, q_0, V, \delta, F)$ be a complete reversible 2-head finite automaton (and as we always assume with only states that are reachable).
We shall prove that if $ \delta(q, (a, b)) = q'$ for some $ a, b \in V $, then $A$ cannot be complete. Assume for contradiction that  
$ A$ has the transition: $\delta(q, (a, b)) = q'$. 

This means that in state $ q $, the first head can read an $ a $ and the second head reads a $ b $ at the same time while  the automaton switches its state to $ q'$.

Let us consider an input $uv$ such that state $q$ is reached from $q_0$ by reading $u$ with the first and $v$ with the second head, respectively. Such input exists by our assumption.
Then, we state that $A$ cannot completely read at least one of the original inputs $uav$ and $ubv$.
The computation on these inputs goes as $(q_0,uav) \Rightarrow^* (q,a)$ and $(q_0,ubv) \Rightarrow^* (q,b)$, respectively. 
Considering the first word, the current configuration is $(q,a)$. To be able to read $a$ in the next step of the computation $A$ must have at least one of the transitions $\delta(q, (a, \lambda)) \ne \emptyset$ or $(q, (\lambda, a)) \ne \emptyset$. 
Actually, as $A$ is deterministic, we can be sure that not both such transitions exist in $A$. Moreover, as we already assumed that
$\delta(q, a, b) \ne \emptyset$, thus
$\delta(q, (a, \lambda)) \ne \emptyset$ is impossible as it would also contradict to the deterministic behavior of $A$.
Thus, we can conclude that to allow to process the input $uav$ completely, $A$ must have 
$(q, (\lambda, a)) \ne \emptyset$.
\\
Now, let us consider the configuration $(q,b)$.
To make the input $ubv$ completely read by $A$, we need to able to read this $b$ at state $q$. For that we would need at least one of 
$\delta(q, (b, \lambda)) \ne \emptyset$ or $\delta(q, (\lambda, b)) \ne \emptyset$. As $A$ is deterministic, in fact, exactly one of those.
However, in case $\delta(q, (b, \lambda)) \ne \emptyset$, we have a contradiction: this with 
$\delta(q, (\lambda, a)) \ne \emptyset$ cannot be in a deterministic 2-head automaton.
\\
Further, the second transition $(q, (\lambda, b)) \ne \emptyset$ cannot be in $A$ as it would contradict to the deterministic behavior of $A$ since it has the transition  $\delta(q, (a, b)) \ne \emptyset$.

In a similar manner, it can also be seen that it is also impossible in a complete reversible automaton  if $a=b$, i.e., the transition in which both heads read is of the form $\delta(q,(a,a)) \ne \emptyset$. 

Therefore, we can conclude that a transition where both heads move requires an input of length greater than $2$. But completeness demands that the automaton must run on any nonempty input, including length $1$. Thus, such a transition cannot be used on all inputs, violating completeness.
\end{proof}

Thus, one may see that completeness automatically involve the
1-limited restriction.  

We have the following structural result about these automata.

\begin{thm} Let  $A$ be a complete 1-limited reversible 2-head automaton. 
    Then the graph of $A$ is strongly connected, i.e., there is a computation for each pair of states $p$ and $q$ such that if the {current} configuration has state $p$, then there is a (remaining) input such that by processing it, $A$ reaches a configuration with state $q$.
\end{thm}
\begin{proof}
    For each state the number of defined transitions is exactly the same as the cardinality of the alphabet, let us say $n$. Since these automata are reversible, there is no state to which more than $n$ transition could be defined. However, in this case, there are exactly $n$ ``incoming'' transitions to each state. Thus, for any subset $Q'$ of states, the number of transitions coming from ``outside'' (i.e., from a state in $Q\setminus Q'$ to a state in $Q'$) must be the same as the number of transitions going outside from $Q'$ (i.e., transitions from a state in $Q'$ to a state in $Q\setminus Q'$). From this it follows that from any state $p$ that is reachable, there is also a computation to the initial state $q_0$.  
\end{proof}

\begin{lem}\label{lem:CCC}
The language $L_{ba}=\{b^na^n, b^{n+1}a^n ~|~ n\geq 0\}$ cannot be accepted by any complete reversible 2-head automaton.
\end{lem}
\begin{proof}
  Let us assume that there is a complete reversible 2-head automaton $A$ that accepts $L_{ba}$. Then, in its initial state $q_0$ either head must start the process. 
  If the first head starts to read, then $A$ must have both transitions with $(a,\lambda)$ and $(b,\lambda)$ to ensure completeness. However, there is no word in $L_{ba}$ that starts with a letter $a$, thus by the former transition a state should be reached from which no accepting computation can be continued. However, this violates the strongly connected property of $A$.
  Thus, $A$ must start with the second head, by having both transitions with $(\lambda,a)$ and $(\lambda,b)$. This second should reach an accepting state, as $b\in L_{ba}$. However, as no other words having suffix $b$ are in the language, from this accepting state each continuation of the computation must not be accepting. As the automaton is complete, it should able to read any unread input, but again, not to accept any continuations would lead to a contradiction to the strongly connectedness of $A$.  
\end{proof} 

\begin{thm}
The class of all languages that can be accepted by reversible complete 2-head automata 
 (RC-2H) is a proper subset of 2revLIN. Furthermore, it is a proper subclass of the class accepted by 1-limited reversible 2-head automata:
 $$2CrevLIN \subsetneq 2rev1LIN.$$
\end{thm}
\begin{proof}
It is clear that RC-2H is a subset of R-2H, and thus 2CrevLIN is a subset of 2revLIN from the definitions. Moreover, RC-2H is also a subset of R1-2H by Proposition \ref{RC&R1}. The strict inclusion comes from Lemma \ref{lem:CCC} as the language $L_{ba}=\{b^na^n, b^{n+1}a^n ~|~ n\geq 0\}$ can clearly be accepted by a 1-restricted, but not complete reversible 2-head automaton with 2 states, let us say $q_0$ and $q$, having a transition from $q_0$ to $q$ by reading a $b$ with the first head, and having a transition from $q$ to $q_0$ by reading an $a$ with the second head. Both states are also final states.  
\end{proof}

Further, about the closure properties of the language class accepted by RC-2H automata we have the following.
\begin{thm}
The class 2CrevLIN of languages accepted by complete (1-limited) 2-head reversible automata is closed under complementation operation.
\end{thm}
\begin{proof}
    As an RC-2H automaton $A=(Q,q_0,V,\delta,F)$ can fully read every input in a deterministic manner, the automaton $\bar{A}=(Q,q_0,V,\delta,Q\setminus F)$ accepts exactly the complement of $L(A)$. Since only the accepting states are changed $\bar{A}$ also belongs to the class RC-2H. 
\end{proof}

We close this section with another closure property that we have for each our new classes. 

\begin{thm}
Each of the classes 2detLIN, 2revLIN, 2rev1LIN and 2CrevLIN of languages, i.e., the classes accepted by deterministic 2-head automata, reversible 2-head automata, by 1-limited reversible 2-head automata and by complete (1-limited) 2-head reversible automata, respectively, is closed under reversal operation.
\end{thm}
\begin{proof}
Observe that the construction used in Lemma \ref{reversal} does not modify any of the following  properties: determinism, reversibility, 1-limitedness and completeness. Therefore, if the original automaton has any of these properties, the automaton that accepts the reversal of the language has also the same restrictions.
\end{proof}

\section{Discussion and Summary}
In this section, we start  
with a relation among the language classes accepted by our 
deterministic/rever\-sible 2-head automata and those that are generated by left  
deterministic linear grammars (Definition~\ref{def-left}).  

\begin{thm}\label{thm:LD-2det}
    Every left deterministic linear language is accepted by deterministic 2-head automata, moreover, the inclusion $LDLL \subsetneq 2detLIN$ is proper.
\end{thm}

We present various example languages in Fig. \ref{fig:linear-languages} (some of the proofs about them are left for the reader). 

\begin{prop} The regular language $L_{ab} = \{a^nb^m\}$ is clearly in LDLL.
\end{prop}

Further,
in this section, we summarize the results concerning the relations between the considered classes on Fig.~\ref{fig:linear-languages}. We have proven the proper hierarchy $$2CrevLIN \subsetneq 2rev1LIN \subsetneq 2revLIN \subsetneq 2detLIN.$$ Moreover, as we have seen, there are regular languages that are not in 2revLIN. However, still 2rev1LIN contains, e.g., the language of palindromes.

\begin{figure}[ht]
    \centering
    \begin{tikzpicture}[scale=0.8, every node/.style={transform shape}]
        
        \draw[blue, thick] (0,0) rectangle (17,10);
         \draw[brown, thick] (0.5,0.5) rectangle (15,9);

        \node[above right, brown] at (8, 9) {\(2\text{detLIN}\)};

\node[above right] at (1.4,7.6) {\(2revLIN\)};
        
        \draw[thick] (4,4.5) ellipse (3cm and 3.5cm); 
                
        \draw[red, thick] (4.5,4) ellipse (2cm and 1.5cm);
        \node[above, red] at (5.2,5.4) {\(2rev1LIN\)};
        
        \draw[green, thick] (3.9,4) ellipse (0.8cm and 0.8cm);
        \node[above, green] at (5.1,4.6) {\(2CrevLIN\)};
        
        \draw[violet,thick] (7,4.5) ellipse (3.5cm and 3.5cm);
        \node[violet,above right] at (9,7.2) {\(LDLL\)};
        
        \node[red] at (5.5,3.5) {\(L_{ba}\)};
        \node[red] at (2.8,4) {?};
        \node[black] at (1.8,4) {?};

         \node[green] at (4.1,4) {\(V^*\)};
         \node[green] at (3.34,4) {?};

        \node[violet] at (8,5) {\(L_{ab}\)};

        \node[orange] at (12,2) {\(L_{ab+ac}\)};
        \node[blue] at (16,7) {\(L_{1\lor3}\)};

        \node[text width=4cm, align=center] at (5.2,2.1) {
            \(L_{wcb^n}\)
        };
        
        \node[above, blue, text width=14cm, align=center] at (9,10) {LIN
        };
        
    \end{tikzpicture}
    \caption{Diagram of sublinear languages with  the example language $V^*$ and other examples from the paper.}
    \label{fig:linear-languages}
\end{figure}
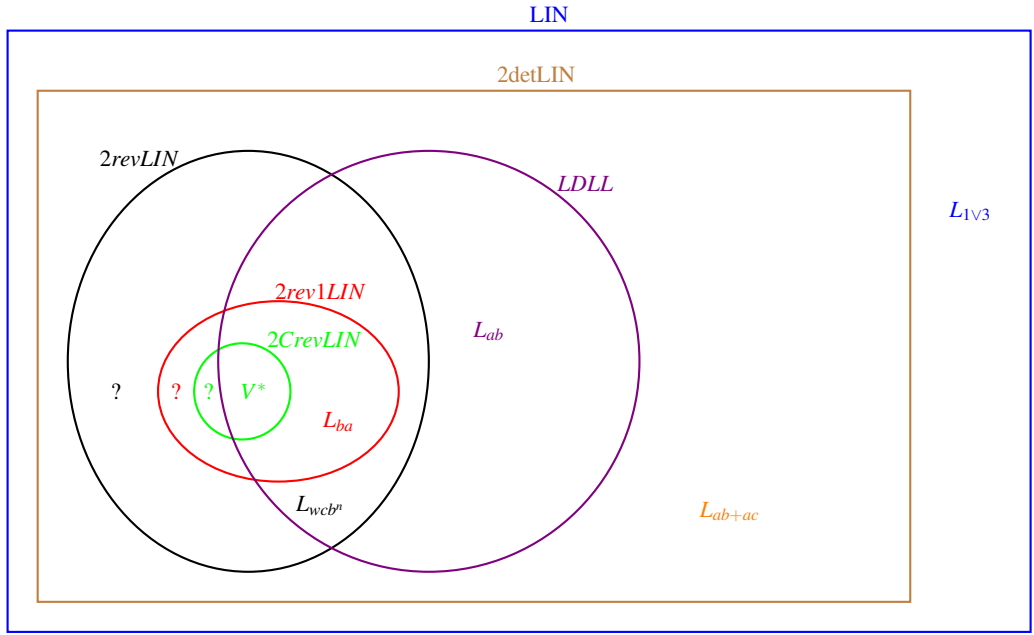

There are some open problems as well. 
There are some regions of the diagram where we do not put any example language (see ? marks in the figure). It is a future task to find examples to fit there, or to prove if some of these regions are empty.
Based on our examples, we conjecture that $2CrevLIN$ contains only regular languages. Decidability problems, computational and descriptional complexity issues regarding the new classes are also
 open and can be studied in the future.

\section*{Acknowledgments} The authors are very grateful to the reviewers for their comments.

\bibliographystyle{eptcs}
\bibliography{References}
\end{document}